\documentclass[12pt]{amsart}
\usepackage{a4wide}
\usepackage[dvips]{epsfig}
\usepackage{amscd}
\usepackage{amssymb}
\usepackage{amsthm}
\usepackage{amsmath}
\usepackage{latexsym}
\usepackage[linktocpage,
                  menucolor=blue,
                  linkcolor=blue,
                  citecolor=red,
                  colorlinks=true]{hyperref}

\theoremstyle{plain}
\newtheorem{theorem}{Theorem}[section]
\theoremstyle{plain}

\newtheorem{corollary}[theorem]{Corollary}
\theoremstyle{definition}
\newtheorem{definition}{Definition}[section]

\newtheorem{assumption}{Assumption}
\newtheorem{example}{Example}[section]
\usepackage[toc,page]{appendix}
{%
\setcounter{enumi}{0}

\begin{enumerate}}%
{\end{enumerate} }
{%
\setcounter{enumi}{0}

\begin{enumerate}}%
{\end{enumerate} }

\numberwithin{equation}{section}
 \allowdisplaybreaks

\title[A note on the representation of BSDE-based dynamic risk measures]{A note on representation of BSDE-based dynamic risk measures and dynamic capital allocations. }

\date{\today}

\begin{document}

\author{Lesedi Mabitsela, Calisto Guambe and Rodwell Kufakunesu}
\address{Department of Mathematics and Applied Mathematics, University of Pretoria, 0002, South Africa}
\address{Department of Mathematics and Informatics, Eduardo Mondlane University, 257, Mozambique}

\email{rodwell.kufakunesu@up.ac.za, calistoguambe@yahoo.com.br, lesedi.mabitsela@up.ac.za}

%
%

\keywords{
 Dynamic risk capital allocation, dynamic risk measure, quadratic-exponential BSDE, dynamic entropic risk measure}

\begin{abstract}
  In this paper, we provide a representation theorem for dynamic capital allocation under It{\^o}-L{\'e}vy model. We consider the representation of dynamic risk measures defined under Backward Stochastic Differential Equations (BSDE) with generators that grow quadratic-exponentially in the control variables. Dynamic capital allocation is derived from the differentiability of BSDEs with jumps. The results are illustrated by deriving a capital allocation representation for dynamic entropic risk measure and static coherent risk measure.
\end{abstract}

\maketitle
\section{Introduction}
Risk measurement is the process of quantifying uncertainty in the future value of a financial position. Risk describes the changes in the future value of a position, due to uncertain events. The functionals used to quantify risk are known as coherent and convex risk measures, as they assign random variables to real numbers. Artzner et al. \cite{artzner1999} proposed coherent risk measure and described it as a function that satisfies four properties: translation invariance, monotonicity, sub-additivity and positive Homogeneity. Delbaen \cite{delbaen2002} defined coherent risk measures in the general probability space. F{\"o}llmer and Schied \cite{follmer2002} and independently Frittelli and Rosazza Gianin \cite{frittelli2002}, extended the work of Artzner et al. \cite{artzner1999} by introducing the concept of convex risk measures. A convex risk measure takes into consideration that the risk of a position may increase in a nonlinear way as a position multiplies by a large factor. Hence, the properties of positive homogeneity and subadditivity are relaxed \cite{follmer2002}.\\

In the abovementioned papers, the authors consider risk measure in a single-period setting. The ideal situation is to measure the risk of a financial position continuously throughout the investment period. Various authors have extended the concept of static risk measure to dynamic risk measure. Peng \cite{Peng1997} introduced $g$-expectations as nonlinear expectations based on BSDEs. Rosazza Gianin \cite{Rosazza2006} showed that conditional $g$-expectations define dynamic risk measures under diffusion BSDE (see also \cite{barrieu2007}, \cite{frittelli2002}, \cite{peng2004}). Jiang \cite{jiang2008}  proved that $g$-expectation satisfies the translation invariance property if and only if $g$ is independent of $y$ and $g$-expectation is convex if and only if $g$ is independent of $y$ and $g$ is convex with respect to $z$. In \cite{quenez2013}, Quenez and Sulem study properties of dynamic risk measures based on BSDEs with jumps (see also \cite{oksendal2015} for applications). An extension of quadratic BSDE to jumps is done by Karoui et al. in \cite{karoui2016} and they include an application to entropic risk measures. There are two representations of risk measures that are mostly applied in literature, namely scenario-based representation and the representation using backwards stochastic differential equations. \\

Risk measures are used to determine the amount required to hold as a buffer against unexpected losses for a portfolio. Risk measures can be further used to measure the risk contribution of a subportfolio in a overall portfolio (see for example \cite{cherny2009}, \cite{buch2008}, \cite{denault2001}, \cite{kalkbrener2005} and \cite{tasche2004}). Capital allocation is the problem of measuring the risk contribution of subportfolio in a portfolio to the overall portfolio risk. The capital allocation methods that are mostly studied are Shapley method, Aumann Shapley and Euler's method also known as the gradient allocation method. Denault in \cite{denault2001} provided the properties of coherent capital allocation. These are coherent capital allocation properties are ``no undercut", symmetry and riskless allocation, which together justify the gradient allocation principal. Denault showed that the Aumman Shapley value is coherent and a practical approach to capital allocation \cite{denault2001}. Kalkbrener in \cite{kalkbrener2005} further provided properties for gradient allocation principle, and shows that the properties are satisfied if and only if the risk measure is positive homogeneous and sub-additive. The gradient allocation properties provide by Denault in \cite{denault2001} are shown to be equivalent to risk measure axioms of positive homogeneity, sub-additivity and translation invariance respectively (see Buch in \cite{buch2008}). The Euler's or gradient allocation method is analysed by Tasche in \cite{tasche2007}. \\

Tasche in \cite{tasche2004} showed that if the risk measure if smooth, then the partial derivative of the risk measure with respect to the underlying asset is the unique gradient allocation principle.  Kromer and Overbeck \cite{kromer2014} used the results of Ankrichner et al. \cite{ankirchner2007} for the differentiability of the BSDE. In this paper, we use the results from Fuji \cite{fujii2017} for the differentiability of quadratic-growth BSDE with jumps to derive the representation of BSDE-based dynamic gradient allocation.\\

This paper is motivated by the representation of BSDE-based dynamic capital allocation from Kromer and Overbeck \cite{kromer2014} and extend it to BSDEs with jumps. The jump-diffusion model can help to explain the extreme movements in a risky asset. Rong in \cite{rong2006} explains the jump term as a random extreme move in financial returns caused by, for example, the announcement of some important economic policy to this financial market or the performance of an important decision made by some big company.\\

The remainder of the paper is organised as follows. In Section 2, we present the notations and define concepts that will be used throughout the paper. Section 3, we derive the representation of dynamic risk capital allocation based on BSDE with jumps. From dynamic risk capital allocation results we derive the representation of BSDE based dynamic convex and coherent risk measures. We conclude in Section 4 with applications of our results to the entropic risk measures.\\

\section{Preliminaries}
In this section we introduce the main concepts and notations to be used throughout the paper. Consider $(\Omega, \mathcal{F}, \mathbb{P})$ to be the filtered probability space with the filtration $\{\mathcal{F}_{t}\}_{t\in[0,T]}$. Let $L^{p}(\Omega, \mathcal{F}, \mathbb{P})$, $1\leq p < \infty$ be the space of all real-valued $\mathcal{F}_t$-measurable, $p$-integrable random variables. Let $L^{\infty}(\Omega, \mathcal{F}, \mathbb{P})$ be the space of essentially bounded random variables with norm $||X||_{\infty}=ess sup|X|$ and $L^0(\Omega,\mathcal{F}_t,\mathbb{P})$ the space of all $\mathcal{F}_t$-measurable random variables. Denote by $\mathcal{X}\subset L^{2}(\Omega, \mathcal{F}, \mathbb{P})$ the space of financial positions $X$.\\

\begin{definition} (see \cite{artzner1999}, \cite{Rosazza2006})
	A mapping $\rho : \mathcal{X} \rightarrow \mathbb{R}$ is a static risk measure if, for any $X$ and $Y$ in $\mathcal{X}$, it satisfies the following axioms:
	\begin{itemize}
		\item [1)] Monotonicity: $\rho(X) \leq \rho(Y)$, $\forall$  $Y\leq X$;
		\item [2)] Translation invariance: $\rho(X+m)=\rho(X)-m$, $m\in\mathbb{R}$;
		\item [3)] Subadditivity: $\rho(X+Y)\leq \rho(X)+\rho(Y)$;
		\item [4)] Positive homogeneity $\rho(k X)=k\rho(X)$, $k\geq 0$;
		\item [5)] Convexity: $\rho(\lambda X +(1-\lambda)Y)\leq \lambda
		\rho(X)+(1-\lambda)\rho(Y)$, $\lambda \in (0,1)$.
	\end{itemize}
    The functional $\rho(X)$ quantifies the risk of a financial position $X\in \mathcal{X}$. A financial position $X\leq 0$ is acceptable if $\rho(X)\geq 0$, and $\rho(X)$ represents the capital amount that an investor can withdraw without changing the acceptability of $X$. Monotonicity implies that $\rho$ is nonincreasing with respect to $X \in \mathcal{X}$. The financial meaning is that if a financial position $X$ is always higher than $Y$, then the capital required to support $X$ should be less than capital required for $Y$. Subadditivity allows for risk to be reduced by diversification since the risk of a portfolio $X+Y$ is bounded by the sum of individual risk of position $X$ and $Y$.
	Translation invariance states that if you add a certain amount $m$ to the initial investment position, then the risk of that investment will decrease by that amount $m$. Note that, if a position $X$ is not acceptable, then adding an amount $\rho(X)$ to it will make the position acceptable, i.e. $\rho(X+\rho(X))=\rho(X)-\rho(X)=0$. Positive homogeneity tells us that the capital required to support $k$ identical positions is equal to $k$ times the capital required for one position.\\
	
 A convex risk measure $\rho$ whose domain includes $\mathcal{X}$ such that $\rho(X)<\infty$ where $X\in \mathcal{X}$, satisfies property 5) see \cite{follmer2002} and \cite{frittelli2002}, while a coherent risk measure satisfies properties 1) to 4) see \cite{artzner1999} and \cite{delbaen2002}. The position $X$ is acceptable when $\rho(X)\leq 0$, and unacceptable otherwise \cite{artzner1999}. We state from (Rosazza Gianin \cite{Rosazza2006}) the following definition of a dynamic risk measure:
\end{definition}
\begin{definition}
	A mapping $(\rho_t)_{t\in[0,T]}$ is a dynamic risk measure for all $X, Y \in \mathcal{X}$ and $t\in[0,T]$, if the following properties are satisfied:
	\begin{itemize}
		\item [(a)] $\rho_t:L^p(\mathcal{F}_T)\rightarrow L^0(\Omega,\mathcal{F}_t,\mathbb{P})$.
		\item [(b)] $\rho_0$ is a static risk measure.
		\item [(c)] $\rho_T(X)=-X$ for all $X\in \mathcal{X}$.
	\end{itemize}
\end{definition}
A dynamic risk measure is called coherent if it satisfies, positive homogeneity, monotonicity, translation invariance and subadditivity. A convex dynamic risk measure satisfies the convexity property and $\rho_t(0)=0$ for any $t\in[0,T]$ (see \cite{Rosazza2006}).\\

Let $X_1, X_2,\ldots,X_n \in \mathbb{R}$ be financial positions, with the corresponding risk contribution to the overall portfolio denoted by $\rho(X_i|X)$, $i=1,2,\ldots,n.$ Consider a portfolio $X\in \mathcal {X}$, consisting of $X_i$, $i=1,2,\ldots,n$ subportfolios, that is
\[
	X=\sum_{i=1}^{n}X_i.
\]

The portfolio risk is given by $\rho(X)$. The capital allocation problem is allocating the overall risk $\rho(X)$ of the portfolio $X$ to the individual subportfolios in the portfolio. That is, we require a mapping such that
\begin{equation}\label{fullAllocationProperty}
	\rho(X)=\sum_{i=1}^{n}\rho(X_i|X).
\end{equation}
Such a relation is called the \textit{full allocation property}, since the overall portfolio risk is fully allocated to the individual subportfolios in the portfolio (see \cite{tasche2007}).
\begin{definition}
	Let $\rho$ be the risk measure that is continuously differentiable at $X$ in its domain, then $\rho(X_i|X)$ is uniquely determined by
	\begin{eqnarray}
		\rho(X_i|X)&=&\nabla_{X_i}\rho(X)\nonumber\\
		&=&\lim\limits_{h\rightarrow 0}\frac{\rho(X+hX_i)-\rho(X)}{h}\nonumber\\
		&=&\frac{d}{dh}\rho(X+hX_i)\bigg|_{h=0}.\label{gradientAllocation}
	\end{eqnarray}
\end{definition}
Note that the gradient of a continuous differentiable risk measure $\rho(X_i|X)$ is the unique allocation principle (see proposition 2.1 in Tasche \cite{tasche2007}). Equation \eqref{gradientAllocation} defines static gradient allocation principle, which is the G\^{a}teaux-derivative of $X$ in the direction of $X_i$, for $i=1,2,\ldots,n$. The static Aumann-Shapley allocation is defined by:
\begin{equation}\label{Aumann}
		\overline{\nabla_{X_i}\rho(X)}=\int_{0}^{1}\nabla_{X_i}\rho(\beta X)d\beta, \qquad i=1,2,\ldots,n.
\end{equation}
If the risk measure $\rho$ is positive homogeneous, then the Aumann-Shapley allocation reduces to the gradient allocation principle \eqref{gradientAllocation} as indicated by Denault \cite{denault2001}. For the Aumann-Shapley, we do not require the risk measure to be positively homogeneous to satisfies full allocation property. While the gradient allocation does need the risk measure to be positive homogeneous to satisfy the full allocation property. According to Kromer and Overbeck \cite{kromer2014}, the Aumann-Shapley and G\^ateaux-derivative can be jointly used to risk measures that do not satisfy the positive homogeneity property. \\

Let $W=\{W(t),\mathcal{F}(t); 0\leq t \leq T\}$ be the one-dimensional standard Brownian motion, defined on a probability space $(\Omega^{W},\mathcal{F}^{W},\mathbb{P}^{W})$. Furthermore, let $\tilde{N}(dt,d\zeta):=N(dt,d\zeta)-\nu(d\zeta)dt$ be the independent compensated Poisson random measure defined on the probability space $(\Omega^{\tilde{N}},\mathcal{F}^{\tilde{N}},\mathbb{P}^{\tilde{N}})$, with $\nu$ on $\mathbb{R}_0=\mathbb{R}\backslash \{0\}$ as the L\'evy measure of $N$. The Poisson random measure $N$, counts the number of jumps of size $\Delta X$ that occur on or before time $t$. The Brownian and Poisson probability space $(\Omega, \mathcal{F}, \mathbb{P})$ is the product of the described probability spaces $(\Omega^{W}\otimes\Omega^{\tilde{N}},\mathcal{F}^{W}\otimes\mathcal{F}^{\tilde{N}},\mathbb{P}^{W}\times\mathbb{P}^{\tilde{N}})$. We consider the complete filtered probability space $(\Omega, \mathcal{F}, (\mathcal{F}_t)_{t\in[0,T]}, \mathbb{P})$ and the following spaces of random and stochastic processes:
\begin{itemize}
	\item $L^2(\mathcal{F}_T)$ is the space of $\mathcal{F}_T$-measurable, square integrable random variable $\xi$.
	\item $\mathbb{S}^2(\mathbb{R})$ is the space of $\mathbb{R}$-valued $Y:\Omega\times [0,T]$ c\`{a}dl\`{a}g processes such that
	\[\mathbb{E}[\sup_{t\in[0,T]}|Y(t)|^2]<\infty.\]
	\item $\mathbb{H}_W^2(\mathbb{R})$ is the space of predictable processes $Z:\Omega\times [0,T]\rightarrow\mathbb{R}$ such that
	\[
	\mathbb{E}[\int_{0}^{T}|Z(s)|^2ds]<\infty.
	\]
	\item $\mathbb{H}^2_N(\mathbb{R})$ denotes the space of predictable processes $\Upsilon:\Omega\times[0,T]\times\mathbb{R}\rightarrow \mathbb{R}$, satisfying
	\[
	\mathbb{E}\big[\int_{0}^{T}\int_{\mathbb{R}_0}|\Upsilon(t,\zeta)|^2\nu(d\zeta)dt\big]<\infty.
	\]
\end{itemize}

Let $X(t)$ be a L\'evy process with a semi-martingale decomposition $X(t)=X(0)+M(t)-V(t)$, where $V$ is the continuous finite variation drift defined by
\[
V(t)=\int_{0}^{t}\bigg[\mu +\frac{\sigma^2}{2}+\int_{|\zeta|<1}(e^{\Upsilon(s,\zeta)}-1-\Upsilon(s,\zeta))\nu(d\zeta)\bigg]ds
\]
 and $M$ is the local martingale given by
\[
M(t)=M(0)+\int_{0}^{t}Z(s)dW(s)+\int_{0}^{t}\int_{\mathbb{R}_0}(e^{\Upsilon(s,\zeta)}-1)\tilde{N}(ds,d\zeta).
\]

 Given a local martingale $M(t)$, $M(0)=0$, then an adapted process $\Gamma(t)$ that has a stochastic differential equation of the form $d\Gamma(t)=\Gamma(t)dM(t)$, $\Gamma(0)=1$ is the stochastic exponential of $M(t)$, denoted by $\Gamma(t) = \mathcal{E}(M)(t)$ and defined as
$$
\mathcal{E}(M(t))=\exp\{M(t)-\frac{1}{2}\langle M^c(t)\rangle\}\times\prod_{0\leq s\leq t}(1+\Delta M^J(s))e^{-\Delta M^J(s)}\,,
$$
where $\langle M \rangle$ denotes a quadratic variation of a process $M$ and $M^c,\,M^J$ are continuous and discontinuous part of $M$, respectively. Moreover, we introduce the notion of martingales of bounded mean oscillation ({\it BMO}-martingales) for jump-diffusion processes as in Morlais \cite{morlais2009}. A local Martingale $M$ is in the class of {\it BMO}-martingales if there exists a constant $K>0$, such that, for all $\mathcal{F}$-stopping times $\mathcal{T}$,
$$
ess\sup_{\Omega}\mathbb{E}[\langle M(T)\rangle-\langle M({\mathcal{T}}) \rangle\mid\mathcal{F}_{\mathcal{T}}]\leq K^2 \ \ \ \ \ {\rm{and}} \ \ \ \ ess\sup_{\Omega}|\Delta M({\mathcal{T}})|\leq K^2\,.
$$
For the diffusion case, the $BMO$-martingale property follows from the first condition, whilst in a jump-diffusion case, we need to ensure the boundedness of the jumps of the local martingale $M$. The Kazamaki's criterion states that if $M$ is a BMO martingale satisfying $\Delta M(t)\geq -1+\delta$, $\mathbb{P}$-a.s, for $0<\delta<\infty$, and for all $t$ then $\mathcal{E}(M)$ is a true martingale \cite{morlais2009}.\\

Let $D_t$ and $D_{t,\zeta}$ be the Malliavin derivatives with respect to $W$ and $\tilde{N}(dt,d\zeta)$ respectively. A random variable $F$ is Malliavin differentiable if and only if $F\in\mathbb{D}^{1,2}\subset L^2(\mathbb{P})$. The respective norm is defined as follows
$$
\|F\|^2_{1,2}:=\mathbb{E}\Bigl[|F|^2+\sum_{i=1}^d\int_0^T|D^i_sF|^2ds+ \sum_{i=1}^k\int_0^T\int_{\mathbb{R}_0}|D^i_{s,\zeta}F|^2\zeta^2\nu_i(d\zeta)ds\Bigl]\,.
$$
We consider a quadratic exponential BSDE, defined as in Karoui et al. \cite{karoui2016} for $t\in[0,T]$ of the form
\begin{equation}\label{BSDE}
Y(t)=\xi+\int_{t}^{T}g(s,Y(s),Z(s),\Upsilon(s,\zeta))ds-\int_{t}^{T}Z(s)dW(s)-\int_{t}^{T}\int_{\mathbb{R}_0}\Upsilon(s,\zeta)\tilde{N}(ds,d\zeta),
\end{equation}
where $\xi:\Omega\rightarrow \mathbb{R}$ and $g:\Omega\times [0,T]\times\mathbb{R}\times\mathbb{R}^d\times\mathbb{R}^k\rightarrow\mathbb{R}$.  Dynamic risk measures are usually constructed using BSDEs, the following definition from Delong \cite{delong201} defines risk measures induced by BSDEs.
\begin{definition}
Let $\rho_t^g(\xi)=Y^{\xi}(t)$, $t\in[0,T]$. Then $\rho$ is monotone, time-consistent dynamic risk measure. In addition,
	\begin{itemize}
		\item [(a)] if g is sublinear in $(z, \Upsilon)$ and independent of $y$, then $\rho$ is a coherent dynamic risk measure.
		\item [(b)] If $g$ is convex in $(y, z, \Upsilon)$, then $\rho$ is a convex dynamic risk measure
	\end{itemize}
\end{definition}
The component $Y^{\xi}$ is the solution of the BSDE \eqref{BSDE}. The driver $g$ plays an essential role in the construction of risk measures induced by BSDE. For this paper, we focus on quadratic-exponential BSDEs. For the existence and uniqueness of such BSDEs, the driver and terminal condition are subject to the following assumptions (see \cite{briand2006}, \cite{karoui2016}, \cite{fujii2017}, \cite{royer2006}, \cite{delong201}).

\begin{assumption}\label{assumption3.1}
	~\\
\begin{itemize}
	\item [(i)] The map $(\omega, t)\mapsto g(\omega,t,\cdot)$ is $\mathcal{F}_t$-progressively measurable. For $(y,z,\Upsilon)\in \mathbb{R}\times\mathbb{R}\times\mathbb{R}$, there exist two constants $\beta\geq0$ and $\alpha>0$ and a positive $\mathcal{F}_t$-progressively  measurable process $(\ell_t,t\in[0,T])$ such that
	\begin{eqnarray}
	-\ell_t-\beta|y|-\frac{\alpha}{2}|z|^2-\int_{\mathbb{R}_0}j_{\alpha}(-\Upsilon(\zeta))\nu(d\zeta)\leq g(t,y,z,\zeta)\nonumber\\
	\leq \ell_t+\beta|y|+\frac{\alpha}{2}|z|^2+\int_{\mathbb{R}_0}j_{\alpha}(\Upsilon(\zeta))\nu(d\zeta)
	\end{eqnarray}
	$dt\otimes d\mathbb{P}$-a.e. $(\omega,t)\in\Omega\times[0,T]$, were $j_{\alpha}:=e^{\alpha \Upsilon}-1-\alpha \Upsilon$.
	\item [(ii)] $|\xi|$, $(\ell_t,t\in[0,T])$ are essentially bounded i.e. $||\xi||_{\infty}, ||\ell||_{L^{\infty}}<\infty$.
\end{itemize}
\end{assumption}
\begin{assumption}\label{assumption3.2}
	For $M>0$ and $(y,z,\Upsilon), (y',z',\Upsilon')\in\mathbb{R}\times\mathbb{R}\times\mathbb{R}$ satisfying
	\[
	|y|,|y'|,||\Upsilon||_{L^{\infty}},||\Upsilon'||_{L^{\infty}}\leq M,
	\]
	there exists some positive constant $K_M$ possibly depending on $M$ such that
	\begin{eqnarray}
	|g(y,z,\Upsilon)-g(y',z',\Upsilon')|\leq K_M(|y-y'|+||\Upsilon-\Upsilon'||_{L^{2}})\nonumber\\
	+K_M(1+|z|+|z'|+||\Upsilon||_{L^{2}}+||\Upsilon'||_{L^{2}})|z-z'|
	\end{eqnarray}
	$dt\otimes d\mathbb{P}$ a.e. $(\omega,t)\in\Omega\times[0,T]$.
\end{assumption}

\begin{assumption}\label{assumption4.1}
	For all $t\in[0,T]$, $M>0$ and $y\in\mathbb{R}, z\in\mathbb{R},\Upsilon,\Upsilon'\in L^2$ with $|y|,||\Upsilon||_{L^{\infty}},||\Upsilon'||_{L^{\infty}}\leq M
	$, there exists a process $\Gamma^{y,z,\Upsilon,\Upsilon'}$ satisfying $dt\otimes d\mathbb{P}$-a.e.
	\begin{eqnarray}
	g(y,z,\Upsilon)-g(y',z',\Upsilon')\leq \int_{\mathbb{R}_0}\Gamma^{y,z,\Upsilon,\Upsilon'}(\zeta)[\Upsilon(\zeta)-\Upsilon'(\zeta)]\nu(d\zeta)
	\end{eqnarray}
	and $C^1_M(1\wedge|\zeta|)\leq\Gamma^{y,z,\Upsilon,\Upsilon'}_t(\zeta)\leq C^2_M(1\wedge|\zeta|)$ with two constants $C^1_M>-1$ and $ C^2_M>0$ depend on $M$.
\end{assumption}

Fujii and Takahashi [\cite{fujii2017} in Theorem 3.1, Assumptions \ref{assumption3.1}, \ref{assumption3.2} and \ref{assumption4.1}] proved the existence of a unique bounded solution $(Y,Z,\Upsilon)\in \mathbb{S}^2\times\mathbb{H}_W^2\times \mathbb{H}_N^2$ of the BSDE \eqref{BSDE}. Moreover, $Z$ belongs to the set of progressively measurable real valued functions denoted by $\mathbb{H}^2_{BMO(W)}$ satisfying
\[
\bigg|\bigg|\int_{0} ^.Z(s)\bigg|\bigg|^2_{BMO(W)}=ess\sup\mathbb{E}\bigg[\int_{\tau}^T|Z(s)|^2ds\big|\mathcal{F}_{\tau}\bigg]\leq K^2 \quad \mathbb{P}-a.s.
\]
and $\Upsilon$ belongs to the set of predictable processes, denoted by $\mathbb{H}^2_{BMO(N)}$  satisfying
\[
\bigg|\bigg|\int_{0}^. \int_{\mathbb{R}_0}\Upsilon(\zeta)\tilde{N}(ds,d\zeta)\bigg|\bigg|^2_{BMO(N)}=ess\sup\mathbb{E}\big[\int_{\tau}^T\int_{\mathbb{R}_0}|\Upsilon(s,\zeta)|^2\nu(d\zeta)ds\big|\mathcal{F}_t\big]+
|\Delta M({\mathcal{T}})|\leq K^2.
\]
\newline
For this paper, we consider the terminal condition $\xi$ of the form $\xi(\beta)=\xi+\beta\eta$, where $\xi,\eta\in\mathbb{L}^{2}(\mathcal{F}_T)$. The generator $g$ is defined as follows
\begin{equation}\label{generator}
g(t,z,\Upsilon(t,\zeta))=\ell(t,z,\Upsilon)+\frac{1}{2}\alpha|z|^2+\frac{1}{\alpha}\int_{\mathbb{R}_0}(e^{\alpha \Upsilon}-1-\alpha \Upsilon)\nu(dz)
\end{equation}
and is a special case of the generator in Assumption 1 (i). To define the gradient allocation, we need the differentiability for BSDE. In the Brownian case, Kromer and Overbeck \cite{kromer2014} used classical differentiability results for BSDEs from Ankrichner et al. \cite{ankirchner2007}. For our purpose, we use the results of Fujii and Takahashi \cite{fujii2017}, who extended the work of Ankrichner et al. \cite{ankirchner2007} to the Malliavin's differentiability of the quadratic-exponential BSDE.\\

As in Fujii and Takahashi \cite{fujii2017}, consider the following quadratic-exponential BSDE:
\begin{eqnarray}\label{BSDequation}
  Y(t) &=& \xi(X_T)-\int_t^T Z(s)dW(s)-\int_t^T\int_{{\mathbb{R}_0}}\Upsilon^\beta(s,\zeta)\tilde{N}(ds,d\zeta)\nonumber \\
   && +\int_t^Tg\bigg(s,Y(s),Z(s),\int_{\mathbb{R}_0}p(\zeta)G(s,\Upsilon(s,\zeta))\nu(d\zeta)\bigg)ds\,.
 \end{eqnarray}
for $t\in[0,T]$ where $\xi:\Omega\rightarrow \mathbb{R}$, $g:\Omega\times [0,T]\times\mathbb{R}\times\mathbb{R}^d\times\mathbb{R}^k\rightarrow\mathbb{R}$, and $p^i:\mathbb{R}\rightarrow\mathbb{R}$, $G^i:[0,T]\times\mathbb{R}\rightarrow\mathbb{R}$ for each $i={1,\ldots,k}$. The driver $g\bigg(t,Y(t),Z(t),\int_{\mathbb{R}_0}p(\zeta)G(t,\Upsilon(t,\zeta))\nu(d\zeta)\bigg)$, satisfies Assumptions \eqref{assumption3.1} and \eqref{assumption4.1}, where the last arguments denotes a $k$-dimensional vector whose $i$-th element is given by $\int_{\mathbb{R}_0}p^i(\zeta)G^i(s,\Upsilon^i(s,\zeta))\nu^i(d\zeta)$.
Fujii and Takahashi \cite{fujii2017} assume that for every $i\in \{1,\ldots,k\}$, the functions $p^i$ and $G^i(t,v)$ are continuous, with  $p^i$ satisfying $\int_{\mathbb{R}_0}|p^i(\zeta)|^2\nu^id(\zeta)<\infty$. The function $G^i(t,v)$ is continuous in both arguments and one-time continuously differentiable with respect to $v$.\\

\begin{assumption} (Fujii and Takahashi \cite{fujii2017})
	Let $u_t=\int_{\mathbb{R}_0}p(\zeta)G(t,\Upsilon(\zeta))\nu(d\zeta)$ and $u'_t=\int_{\mathbb{R}_0}p(\zeta)G(t,\Upsilon'(\zeta))\nu(d\zeta)$.
	\item [(i)] The terminal value is Malliavin differentiable; $\xi\in\mathbb{D}^{1,2}$.
	\item [(ii)] For
	$M>0$ and $(y,z,\Upsilon)\in\mathbb{R}\times\mathbb{R}^d\times\mathbb{R}^k$ satisfying $|y|,||\Upsilon||_{L^{\infty}}\leq M$, the driver $g(t,y,z,u_t),\, t\in[0,T]$ belongs to $L^{1,2}(\mathbb{R})$ and its Malliavin derivatives is denoted by $D_{t,\zeta} g(t,y,z,u_t)$. Furthermore, the driver $g$ is continuously differentiable with respect to its state variables.
	\item [(iii)] For $M>0$ and $(y,z,\Upsilon), (y',z',\Upsilon')\in\mathbb{R}\times\mathbb{R}^d\times\mathbb{R}^k$, satisfying $|y|,|y'|,||\Upsilon||_{L^{\infty}}, ||\Upsilon'||_{L^{\infty}}\leq M$, the Malliavin derivative of the driver satisfies the following local Lipschitz conditions;
	\[
		|D^i_{t}g(t,y,z,u_t)-D^i_{t}g(t,y',z',u'_t)|\leq K^{M,i}_{t}(|y-y'|+|u_t-u'_t|+(1+|z|+|z'|+|u_t|+|u'_t|)|z-z'|)
	\]
	for $dt$-a.e. $t\in[0,T]$ with $i\in 1,\ldots,d$, and
		\[
		|D^i_{t,\zeta}g(t,y,z,u_t)-D^i_{t}g(t,y',z',u'_t)|\leq K^{M,i}_{t,\zeta}(|y-y'|+|u_t-u'_t|+(1+|z|+|z'|+|u_t|+|u'_t|)|z-z'|)
		\]
	for $dt$-a.e. $t\in[0,T]$ with $i\in 1,\ldots,k$. For ever $M>0$ and $(t,\zeta)$, $(K^{M,i}_{s}(t), t\in[0,T])_{i\in1,\ldots,d}$ and $(K^{M,i}_{s,\zeta}(t), t\in[0,T])_{i\in1,\ldots,d}$ are $\mathbb{R}_+$-valued $\mathcal{F}_t$-progressively measurable processes.	
	\item [(iv)] There exists some positive constant $r\geq 2$ such that
	\[
	\int_{[0,T]\times \mathbb{R}^k}\bigg(\mathbb{E}\big[|D_{t,\zeta}\xi|^{rq}+\big(\int_0^T|D_{t,\zeta} g(s,0)|ds\big)^{rq}+||K^M||^{2rq}\big]\bigg)^{\frac{1}{q}}\tilde{N}(ds,d\zeta)<\infty
	\]
	hold for $\forall q\geq 1$ and $\forall M>0$.
\end{assumption}

Fujii and Takahashi (\cite{fujii2017}, Theorem 5.1), proved that under the above assumptions the solution $(Y,Z,\Upsilon)\in\mathbb{S}^2\times\mathbb{H}^2_{BMO(W)}\times\mathbb{H}^2_{BMO(N)}$ of the BSDE \eqref{BSDequation} is Malliavin differentiable and it is the unique solution to the BSDE
\begin{eqnarray}
D_tY(t)&=&\partial_x\xi(X(T))D_tX(T)-\int_t^TD_{t}Z(s)dW(s)-\int_t^T\int_{\mathbb{R}_0}D_t\Upsilon(s,\zeta)\tilde{N}(ds,d\zeta)\nonumber\\
&&+\int_t^T\bigg[D_tg(s,\Theta)+\partial_yg(s,\Theta)D_tY(s)+\partial_zg(s,\Theta)D_tZ(s)\nonumber\\
&&+\partial_ug(s,\Theta)\int_{\mathbb{R}_0}p(\zeta)\partial_vG(s,\Upsilon(\zeta))D_t(\Upsilon(\zeta))\nu(d\zeta)\bigg]ds,
\end{eqnarray}
where $\Theta:=(Y(t),Z(t),\int_{\mathbb{R}_0}p(\zeta)G(s,\Upsilon(s,\zeta))\nu(d\zeta))$, $t\in[0,T]$. The solution also satisfies $\int_0^T||D_tY,D_tZ,D_t\Upsilon||^rds <\infty$.\\

Let $\xi=\sum_{i=1}^n\eta_i$, then the BSDE version of the dynamic gradient allocation is defined as the directional derivative as follows:
$$
 D_tY(t)= \nabla_{\eta_i}Y(t):=\frac{d}{d\beta}\rho_t(\xi+\beta\eta_i)\Bigl|_{\beta=0}\,,
$$
where
\begin{eqnarray}\nonumber
  D_tY(t) &=& 
  -\eta_i- \int_t^TD_tZ(s)dW(s) -\int_t^T\int_{\mathbb{R}_0}D_t\Upsilon(s,\zeta)\tilde{N}(ds,d\zeta) \\ \label{gradientbsde}
   &&+\int_t^T\bigg[D_tg(s,\Theta)+\partial_zg(s,\Theta)D_tZ(s)\nonumber\\
   &&+\partial_ug(s,\Theta)\int_{\mathbb{R}_0}p(\zeta)\partial_{\upsilon}G(s,\Upsilon(\zeta))D_t\Upsilon(s,\zeta)\nu(d\zeta)\bigg]ds\,.
\end{eqnarray}
 We are now is a position to provide the main result on the representation of the dynamic risk capital allocations as a dynamic gradient allocation.\\

\section{Representation of dynamic risk capital allocations}
In this section, we derive the dynamic risk capital allocation induced from BSDEs with jumps. We also obtain the representation of BSDE based dynamic convex and dynamic coherent risk measures.  We follow the approach of Kromer and Overbeck \cite{kromer2014} in deriving the representation of capital allocation, BSDE based dynamic convex and coherent risk measures.

\begin{theorem}
Let $\xi,\eta_i\in\mathbb{L}^{\infty}(\mathcal{F}_{T})$, such that $\xi=\sum_{i=1}^n\eta_i$ for each $i=1,2,\ldots,n$ and $\nabla_{\eta_i}Y(t)$ exists. Suppose that $\partial_zg(t,\Theta)$ and $\partial_vg(s,\Theta)$ belong to BMO$(\mathbb{P})$. Then the dynamic gradient allocations can be represented by:
$$
\nabla_{\eta_i}Y(t)=\nabla_{\eta_i}\rho_t(\xi)=\mathbb{E}^{\mathbb{Q}^{\xi}}[-\eta_i\mid\mathcal{F}_t]\,, \ \ \ n=1,2,\ldots,n\,,
$$
where $\mathbb{Q}^{\xi}$ is given by
\begin{equation}\label{girsanovtype}
    \frac{d\mathbb{Q}^{\xi}}{d\mathbb{P}}:=\mathcal{E}\Bigl(\int_0^t\partial_zg(s,\Theta)dW +\int_0^t\int_{\mathbb{R}_0}\partial_vg(s,\Theta)\tilde{N}(ds,d\zeta)\Bigl)(t)\,.
\end{equation}

\end{theorem}

\begin{proof}
Since $\partial_zg(t,\Theta)$ and $\partial_{\upsilon}g(t,\Theta)$ belong to BMO$(\mathbb{P})$, then the stochastic integrals in \eqref{girsanovtype} are said to be BMO$(\mathbb{P})$-martingales and the stochastic exponential is a true martingale \cite{morlais2009}. From \cite{delong201} Theorem 2.5.1, a new equivalent probability measure $\mathbb{Q}$ is defined by equation \eqref{girsanovtype}. Furthermore, the processes
\[W^{\mathbb{Q}}=W-\int_0^t\partial_zg(s,\Theta)ds\] and \[\tilde{N}^{\mathbb{Q}}(dt,d\zeta)=N(dt,d\zeta)+\partial_{\upsilon}g(t,\Theta)\nu(d\zeta)dt\]
are the $\mathbb{Q}$- Brownian motion and $\mathbb{Q}$-compensated random measure respectively.
We define a function $\Phi_i(t)$ by $\Phi_i(t):=\mathbb{E}^{\mathbb{Q}^{\xi}}[-\eta_i\mid\mathcal{F}_t]$ for each $i=1,\ldots,n$. Let $Z^{\eta_i}(t)$ and $\Upsilon^{\eta_i}(t,\zeta)$ be predictable processes, then from the martingale representation theorem we have,
\begin{eqnarray}\nonumber
  \Phi_i(t) &=& -\eta_i-\int_t^TZ^{\eta_i}(s)dW^{\mathbb{Q}}-\int_t^T\int_{\mathbb{R}_0}\Upsilon^{\eta_i}(s,\zeta)\tilde{N}^{\mathbb{Q}}(ds,d\zeta) \nonumber\\ \label{martingalephi}
   &=& -\eta_i -\int_t^TZ^{\eta_i}(s)dW+\int_t^TZ^{\eta_i}(s)\partial_zg(s,\Theta)ds-\int_t^T\int_{\mathbb{R}_0}\Upsilon^{\eta_i}(s,\zeta)\tilde{N}(ds,d\zeta) \nonumber\\
   &&-\int_t^T\int_{\mathbb{R}_0}\partial_{\upsilon}g(s,\Theta)\Upsilon^{\eta_i}(s,\zeta)\nu(d\zeta)ds\,.\nonumber
\end{eqnarray}
Comparing the above equation with the BSDE representing the gradient allocation \eqref{gradientbsde}, provided that \eqref{gradientbsde} has a unique solution, we can conclude that the dynamic gradient allocation has the representation
\begin{equation}
D_tY(t)-\int^T_tD_tg(s,\Theta)ds=\nabla_{\eta_i}\rho_t(\xi)=\mathbb{E}^{\mathbb{Q}^{\xi}}[-\eta_i\mid\mathcal{F}_t]\,, \ \ \ i=1,2,\ldots,n\,,
\end{equation}
where ${\mathbb{Q}^{\xi}}$ is from equation \eqref{girsanovtype}.
\end{proof}

From the above theorem, we can immediately obtain the representation result for BSDE based dynamic convex and dynamic coherent risk measures. The results of the representation of BSDE based dynamic convex and coherent risk measures are established from the full allocation property of the Aumann-Shapley allocation (the static case given in Equation \eqref{Aumann}) \cite{kromer2014}.

\begin{corollary}\label{Corollary32}
Let $\xi\in\mathbb{L}^{\infty}(\mathcal{F}_T)$. Suppose that $\ell$ is convex in $z$ and $\upsilon$ and $\partial_zg(t,Z^{\beta\xi}(t),\Upsilon^{\beta\xi}(t,\cdot))$, $\partial_{\upsilon}g(t,Z^{\beta\xi}(t),\Upsilon^{\beta\xi}(t,\cdot))$ belong to the class of $BMO(\mathbb{P})$, for any $\beta\in[0,1]$, where $Z^{\beta\xi}(t),\,\Upsilon^{\beta\xi}(t,\cdot)$ are the controls to the quadratic-exponential BSDE \eqref{BSDequation}, with terminal condition $\rho_{t,\beta}(\xi)=-\beta\xi$. Then, the corresponding quadratic-exponential BSDE-based dynamic convex risk measure can be represented by
$$
\rho_t(\xi)=\mathbb{E}[-\Lambda^{\xi}(T,t)\xi\mid\mathcal{F}_t]\,,
$$
where
\begin{equation}\label{Lambda}
\Lambda^{\xi}(T,t)=\int_0^1\frac{\mathcal{E}(M^{\beta\xi}(T))}{\mathcal{E}(M^{\beta\xi}(t))}d\beta\,, \ \ \ \forall t\in[0,T]\,,
\end{equation}
for $M^{\beta\xi}$ defined by
$$
M^{\beta\xi}(t)=\int_0^t\partial_zg(s,Z^{\beta\xi}(s),\Upsilon^{\beta\xi}(s,\zeta))dW(s) +\int_0^t\int_{\mathbb{R}_0}\partial_\upsilon g(s,Z^{\beta\xi}(s),\Upsilon^{\beta\xi}(s,\zeta))\tilde{N}(ds,d\zeta)\,.
$$
\end{corollary}

\begin{proof}
Following Kromer and Overbeck \cite{kromer2014}, we consider the following
\begin{eqnarray*}
  \rho_t(\xi) &=& \rho_t(1\xi)-\rho_t(0\xi)=\int_0^1\frac{d}{d\beta}\rho_t(\beta\xi)d\beta \\
   &=& \int_0^1\lim_{\epsilon\rightarrow0}\frac{\rho_t((\beta+\epsilon)\xi)-\rho_t(\beta\xi)}{\epsilon}d\beta \\
   &=&  \int_0^1\nabla_\xi\rho_t(\beta\xi)d\beta\,.
\end{eqnarray*}
From the previous theorem, we have
\begin{equation}\label{dynamicriskbeta}
\rho_t(\xi)=\int_0^1\mathbb{E}^{\mathbb{Q}^{\beta\xi}}[-\xi\mid\mathcal{F}_t]d\beta\,.
\end{equation}
Then since $\xi\in\mathbb{L}^{\infty}(\mathcal{F}_T)$, $\mathbb{Q}^{\beta\xi}$ is an equivalent probability measure, $\forall\beta\in[0,1]$. Hence $\xi$ is $\mathbb{Q}^{\beta\xi}$-a.s. bounded. This implies that $\int_0^1\mathbb{E}^{\mathbb{Q}^{\beta\xi}}[-\xi\mid\mathcal{F}_t]d\beta<\infty$. Define $\Lambda^{\xi}(t)=\mathcal{E}(M^{\beta\xi}(t))$. Then, \eqref{dynamicriskbeta} can be written by
\begin{eqnarray*}
\rho_t(\xi)&=&\int_0^1\mathbb{E}^{\mathbb{Q}^{\beta\xi}}[-\xi\mid\mathcal{F}_t]d\beta = \int_0^1\frac{1}{\Lambda^{\beta\xi}(t)}\mathbb{E}^{\mathbb{P}^{\beta\xi}}[-\Lambda^{\beta\xi}(T)\xi\mid\mathcal{F}_t]d\beta\nonumber\\ &=&\mathbb{E}\Bigl[-\Bigl(\int_0^1\frac{\Lambda^{\beta\xi}(T)}{\Lambda^{\beta\xi}(t)}d\beta\Bigl)\xi\mid\mathcal{F}_t\Bigl]\\
&=&\mathbb{E}[-\Lambda^{\xi}(T,t)\xi\mid\mathcal{F}_t],\\
\end{eqnarray*}
which completes the proof.
\end{proof}
\begin{corollary}
		Let $\xi\in\mathbb{L}^{\infty}(\mathcal{F}_T)$. Suppose that $g$ is of the form $g(t,z,\upsilon)=\ell(t,z,\upsilon)$ is convex and positively homogeneous in both $z$ and $\upsilon$. Moreover, suppose that $\partial_z\ell(t,Z^{\beta\xi}(t),\Upsilon^{\beta\xi}(t,\cdot))$, $\partial_{\upsilon}\ell(t,Z^{\beta\xi}(t),\Upsilon^{\beta\xi}(t,\cdot))$ belong to the class of $BMO(\mathbb{P})$, for any $\beta\in[0,1]$. Then, the corresponding BSDE-based dynamic coherent risk measure can be represented by
		$$
		\rho_t(\xi)=\mathbb{E}^{\mathbb{Q}}[-\xi\mid\mathcal{F}_t]\,,
		$$
		where the $\mathbb{Q}$-measure is given by
		\begin{eqnarray}\label{Qmeasure}
		\frac{d\mathbb{Q}}{d\mathbb{P}}\bigg|_{\mathcal{F}_t}&=&\exp\bigg\{-\int^t_0\partial_z\ell(t,Z^{\xi}(t),\Upsilon^{\xi}(t,\zeta))dW-\frac{1}{2}\int^t_0\partial_z\ell(t,Z^{\xi}(t),\Upsilon^{\xi}(t,\zeta))^2ds\nonumber\\
		&&+\int^t_0\int_{\mathbb{R}_0}\bigg(\ln\big(1-\partial_{\upsilon}\ell(t,Z^{\xi}(t),\Upsilon^{\xi}(t,\zeta))\big)+\partial_{\upsilon}\ell(t,Z^{\xi}(t),\Upsilon^{\xi}(t,\zeta))\bigg)\nu(d\zeta)ds\nonumber\\
		&&+\int^t_0\int_{\mathbb{R}_0}\ln\big(1-\partial_{\upsilon}\ell(t,Z^{\xi}(t),\Upsilon^{\xi}(t,\zeta))\big)\tilde{N}(ds,d\zeta)\bigg\}.\,
		\end{eqnarray}
\end{corollary}
\begin{proof}
	From Corollary \ref{Corollary32}, we have the following representation
	$$
	\rho_t(\xi)=\mathbb{E}[-\Lambda^{\xi}(T,t)\xi\mid\mathcal{F}_t],
	$$
	with $\Lambda$ defined in \eqref{Lambda}. Given that $g(t,z,\upsilon)=\ell(t,z,\upsilon)$ and $\ell$ is convex and positively homogeneous, this implies that the corresponding BSDE-based dynamic risk measure $\rho(\cdot)$ satisfies
	$$
	Y^{\beta\xi}(t)=\rho_t(\beta \xi)=\beta\rho_t(\xi)=\beta Y^{\xi}(t) \quad dt\otimes d\mathbb{P}-a.s.
	$$
	for $c>0$ and $0\leq t \leq T$. We show this by considering two BSDEs given by
	\begin{eqnarray*}
	 Y^{\beta \xi}(t) &=& -\beta\xi-\int_t^T Z^{\beta \xi}(s)dW(s)-\int_t^T\int_{{\mathbb{R}_0}}\Upsilon^{\beta \xi}(s,\zeta)\tilde{N}(ds,d\zeta)\nonumber \\
	 && +\int_t^Tg\big(s,Z^{\beta \xi}(s),\Upsilon^{\beta \xi}(s,\zeta))\big)ds\,,
	 \end{eqnarray*}
	 and
	 \begin{eqnarray*}
	 Y^{\xi}(t) &=& -\xi-\int_t^T Z^{\xi}(s)dW(s)-\int_t^T\int_{{\mathbb{R}_0}}\Upsilon^{\xi}(s,\zeta)\tilde{N}(ds,d\zeta)\nonumber \\
	 && +\int_t^Tg\big(s,Z^{\xi}(s),\Upsilon^{\xi}(s,\zeta))\big)ds\,.
	 \end{eqnarray*}
	Then, from the proof of Proposition 6.2.3(b) in Delong \cite{delong201} we conclude that $Y^{\beta \xi}(t)=\beta Y^{\xi}(t)$, $Z^{\beta \xi}(t)=\beta Z^{\xi}(t)$ and $\Upsilon^{\beta \xi}(t,\zeta)=\beta \Upsilon^{\xi}(t,\zeta)$. \\
	
The above results imply that for the representation of the BSDE coherent risk measure,  the process $\mathcal{E}(M^{\beta \xi}(t))(\cdot)$ appearing in \eqref{Lambda} becomes
\[
	\mathcal{E}\bigg(\int_0^t\partial_zg(s,Z^{\beta\xi}(s),\Upsilon^{\beta\xi}(s,\zeta))dW(s) +\int_0^t\int_{\mathbb{R}_0}\partial_\upsilon g(s,Z^{\beta\xi}(s),\Upsilon^{\beta\xi}(s,\zeta))\tilde{N}(ds,d\zeta)\bigg)(t)\]
	\begin{eqnarray*}
	&=&\exp\bigg\{-\int^t_0\partial_zg(t,\beta Z^{\xi}(t),\beta \Upsilon^{\xi}(t,\zeta))dW-\frac{1}{2}\int^t_0\partial_zg(t,\beta Z^{\xi}(t),\beta \Upsilon^{\xi}(t,\zeta))^2ds\nonumber\\
	&&+\int^t_0\int_{\mathbb{R}_0}\bigg(\ln\big(1-\partial_{\upsilon}g(t,\beta Z^{\xi}(t),\beta \Upsilon^{\xi}(t,\zeta))\big)+\partial_{\upsilon}g(t,\beta Z^{\xi}(t),\beta \Upsilon^{\xi}(t,\zeta))\bigg)\nu(d\zeta)ds\nonumber\\
	&&+\int^t_0\int_{\mathbb{R}_0}\ln\big(1-\partial_{ \upsilon}g(t,\beta Z^{\xi}(t),\beta \Upsilon^{\xi}(t,\zeta))\big)\tilde{N}(ds,d\zeta)\bigg\}\,,\nonumber\\
	&=&\exp\bigg\{-\int^t_0\partial_zg(t,Z^{\xi}(t),\Upsilon^{\xi}(t,\zeta))dW-\frac{1}{2}\int^t_0\partial_zg(t,Z^{\xi}(t),\Upsilon^{\xi}(t,\zeta))^2ds\nonumber\\
	&&+\int^t_0\int_{\mathbb{R}_0}\bigg(\ln\big(1-\partial_{\upsilon}g(t,Z^{\xi}(t),\Upsilon^{\xi}(t,\zeta))\big)+\partial_{\upsilon}g(t,Z^{\xi}(t),\Upsilon^{\xi}(t,\zeta))\bigg)\nu(d\zeta)ds\nonumber\\
	&&+\int^t_0\int_{\mathbb{R}_0}\ln\big(1-\partial_{\upsilon}g(t,Z^{\xi}(t),\Upsilon^{\xi}(t,\zeta))\big)\tilde{N}(ds,d\zeta)\bigg\}\,\nonumber\\
	&=&	\mathcal{E}\bigg(\int_0^t\partial_zg(s,Z^{\xi}(s),\Upsilon^{\xi}(s,\zeta))dW(s) +\int_0^t\int_{\mathbb{R}_0}\partial_\upsilon g(s,Z^{\xi}(s),\Upsilon^{\xi}(s,\zeta))\tilde{N}(ds,d\zeta)\bigg)(t)\nonumber\\
	&=&\mathcal{E}(M^{\xi}(t))\,,\\
	\end{eqnarray*}
	because of the positive homogeneity of $g$ in $z$ and $\Upsilon$. In the case of dynamic coherent risk measure, $\Lambda^{\xi}(T,t)$ is given by
	\begin{equation}\label{exponentialM}
		\Lambda^{\xi}(T,t)=\int_0^1\frac{\mathcal{E}(M^{\beta\xi}(T))}{\mathcal{E}(M^{\beta\xi}(t))}d\beta=\frac{\mathcal{E}(M^{\xi}(T))}{\mathcal{E}(M^{\xi}(t))}.
	\end{equation}
	Hence, the BSDE-based coherent risk measure is given by
		$$
		\rho_t(\xi)=\mathbb{E}^{\mathbb{Q}}[-\xi\mid\mathcal{F}_t]\,,
		$$
		where the $\mathbb{Q}$-measure is defined in \eqref{Qmeasure}.\\
\end{proof}
We obtain similar results as Kromer and Overbeck \cite{kromer2014} were the exponential martingale of the BSDE based convex risk measure is dependent on all portfolio weights $\beta\in[0,1]$. The representation of the coherent risk measure is dependent only on $\beta=1$. The difference between these two risk representations is emphasized in Equation \eqref{exponentialM}.
\section{Example}
In this section we apply the methods presented early to dynamic entropic risk measure and static coherent entropic risk measures.
\begin{example}
We consider the well known dynamic entropic risk measure given by
$$
\rho_{t}(\xi)=\frac{1}{\gamma}\ln\mathbb{E}\Bigl[e^{-\gamma\xi}\mid\mathcal{F}_t\Bigl]\,, \ \ \ \gamma>0, \,t\in[0,T]\,.
$$
This example was also considered in \cite{kromer2014}. It has been proved that the above entropic measure is a unique solution of the so called canonical quadratic-exponential BSDE $(g,\xi)$ of the form (See EL Karoui {\it et. al.} \cite{karoui2016})
\begin{eqnarray*}
  \rho_{t}(\xi) &=& -\xi+\int_t^T\Bigl(\frac{\gamma}{2}|Z^\xi(s)|^2+\frac{1}{\gamma}\int_{\mathbb{R}_0} \Bigl(\exp(\Upsilon^\xi(s,\zeta))-\gamma\Upsilon^\xi(s,\zeta)-1\Bigl)\nu(d\zeta)\Bigl)ds \\
   && -\int_t^TZ^\xi(s)dW(s)-\int_t^T\int_{\mathbb{R}_0}\Upsilon^\xi(s,\zeta)\tilde{N}(ds,d\zeta)\,.
\end{eqnarray*}
Note that the generator is given by
$$
g(t,Z,\Upsilon(\zeta))= \frac{\gamma}{2}|Z|^2+\frac{1}{\gamma}\int_{\mathbb{R}_0} \Bigl(\exp(\Upsilon(\zeta))-\gamma\Upsilon(\zeta)-1\Bigl)\nu(d\zeta)\,.
$$
Then we have $\partial_zg(t,Z,\Upsilon(\zeta))=\gamma Z$, $\partial_\upsilon g(t,Z,\Upsilon(\zeta))= \int_{\mathbb{R}_0} \Bigl(\exp(\Upsilon(\zeta))-1\Bigl)\nu(d\zeta)$.\\

Suppose that $\xi$ is from a class of smooth functions such that $\|\xi\|_{1,2}$ exists and is finite. We define a function $\varphi(\xi)=e^{-\gamma\xi}$. Then from the boundedness of $\xi$, we have that $\varphi(\xi)$ is Malliavin differentiable and the generalized Clark-Ocone formula (Di Nunno  et al. \cite{di2009}, Theorem 12.20)
\begin{eqnarray*}
  e^{-\gamma\beta\xi} &=& \mathbb{E}[e^{-\gamma\beta\xi}] +\int_0^T\mathbb{E}[D_t(e^{-\gamma\beta\xi})\mid\mathcal{F}_t]dW(t) +\int_0^T\int_{\mathbb{R}_0}\mathbb{E}[D_{t,\zeta}(e^{-\gamma\beta\xi})\mid\mathcal{F}_t]\tilde{N}(dt,d\zeta)\,.
\end{eqnarray*}
Define
$
\Gamma^{\beta\xi}(t):= \mathbb{E}[e^{-\gamma\beta\xi}\mid\mathcal{F}_t]\,.
$
Then
\begin{eqnarray*}
  \Gamma^{\beta\xi}(t) &=& \Gamma^{\beta\xi}(0) +\int_0^t\mathbb{E}[-\gamma\beta e^{-\gamma\beta\xi} D_s(\xi)\mid\mathcal{F}_s]dW(s) \\
  && +\int_0^T\int_{\mathbb{R}_0}\mathbb{E}[-\gamma\beta e^{-\gamma\beta\xi} D_{s,\zeta}(\xi)\mid\mathcal{F}_s]\tilde{N}(ds,d\zeta) \\
   &=&  \Gamma^{\beta\xi}(0) +\gamma\int_0^t\Gamma^{\beta\xi}(s)Z^{\beta\xi}(s)dW(s) +\gamma\int_0^t\int_{\mathbb{R}_0}\Gamma^{\beta\xi}(s)\Upsilon^{\beta\xi}(s,\zeta)\tilde{N}(ds,d\zeta)\,,
\end{eqnarray*}
where
$$
Z^{\beta\xi}(t)=\frac{-\beta\mathbb{E}[ e^{-\gamma\beta\xi} D_t(\xi)\mid\mathcal{F}_t]}{ \mathbb{E}[e^{-\gamma\beta\xi}\mid\mathcal{F}_t]} \ \ \ \rm{and} \ \ \ \Upsilon^{\beta\xi}(s,\zeta)=\frac{-\beta\mathbb{E}[ e^{-\gamma\beta\xi} D_{t,\zeta}(\xi)\mid\mathcal{F}_t]}{ \mathbb{E}[e^{-\gamma\beta\xi}\mid\mathcal{F}_t]}\,.
$$
Therefore, $\Gamma^{\beta\xi}(t)$ satisfies the following
\begin{eqnarray*}
  \Gamma^{\beta\xi}(t) &=& \Gamma^{\beta\xi}(0)\exp\Bigl\{\gamma\int_0^tZ^{\beta\xi}(s)dW(s) -\frac{\gamma^2}{2}\int_0^t|Z^{\beta\xi}(s)|^2ds \\
   && \ \ \ \ \ \ \ \ \ \ \ \ \ \   +\int_0^t\int_{\mathbb{R}_0}[\ln(1+\gamma\Upsilon^{\beta\xi}(s,\zeta))-\gamma\Upsilon^{\beta\xi}(s,\zeta)]\nu(d\zeta)dt \\
   && \ \ \ \ \ \ \ \ \ \ \ \ \ \ +\int_0^t\int_{\mathbb{R}_0}\ln(1+\gamma\Upsilon^{\beta\xi}(s,\zeta))\tilde{N}(ds,d\zeta)\Bigl\}\,.
\end{eqnarray*}
Hence, the process $\Gamma^{\beta\xi}(t)/\Gamma^{\beta\xi}(0)$ corresponds to the stochastic exponential $\mathcal{E}$ to the process $M^{\beta\xi}(t)$ defined by
$$
M^{\beta\xi}(t)=\int_0^t\partial_zg(s,Z^{\beta\xi}(s),\Upsilon^{\beta\xi}(s,\zeta))dW(s) +\int_0^t\int_{\mathbb{R}_0}\partial_\upsilon g(s,Z^{\beta\xi}(s),\Upsilon^{\beta\xi}(s,\zeta))\tilde{N}(ds,dz)\,,
$$
for $t\in[0,T]$.\\
\end{example}

\begin{example}
In the second example we consider the static entropic coherent risk measure at level $c$ defined by Follmer and Knispel in Definition 3.1 \cite{follmer2011} as follow
\begin{equation}\label{coherententropic}
 \rho(\xi)=\inf_{\gamma >0}\Bigl(\frac{c}{\gamma}+\frac{1}{\gamma}\ln\mathbb{E}\Bigl[e^{-\gamma\xi}\Bigl]\Bigl)
\end{equation}
for $c>0$. From Proposition 3.1 in Follmer and Knispel \cite{follmer2011} there exists a unique $\gamma_c>0$ such that $c=\mathbb{E}^{\mathbb{Q}}[\int\frac{d\mathbb{Q}}{d\mathbb{P}}\ln(\frac{d\mathbb{Q}}{d\mathbb{P}})]$ and the infimum of Equation \eqref{coherententropic} is attained, i.e.
$$
 \rho(\xi)=\frac{c}{\gamma_c}+\frac{1}{\gamma_c}\ln\mathbb{E}\Bigl[e^{-\gamma_c\xi}\Bigl].
$$
The G\^{a}teaux-differentiable of $\rho$ is given by

\begin{equation}
 \nabla\rho(\beta\xi)=-\frac{e^{-\gamma_c\beta \xi}}{\mathbb{E}[e^{-\gamma_c\beta \xi}]}\,.
\end{equation}
Since the entropic coherent risk measure satisfies the property of positively homogeneity \cite{follmer2011}. Then $\nabla\rho(\beta\xi)=\nabla\rho(\xi)$ and $\Lambda^{\beta\xi}$ for this case will be given by

$$
\Lambda^{\beta\xi}=\int_0^1\nabla\rho(\xi)d\beta=-\int_0^1\frac{e^{-\gamma_c \xi}}{\mathbb{E}[e^{-\gamma_c \xi}]}d\beta,
$$
and therefore $\rho$ from Equation \eqref{coherententropic} can be represented by
$$
\rho(\xi)=\mathbb{E}\Bigl[-\bigg(\int_0^1\frac{e^{-\gamma_c \xi}}{\mathbb{E}[e^{-\gamma_c \xi}]}d\beta\bigg) \xi\Bigl]\,.
$$\\
\end{example}
\section*{acknowledgements}
We would like to thank the following sponsors University of Pretoria, Simons Africa, NRF, TDG-NCP, DST-Risk and the MCTESTP Mozambique for their support.
	\appendix
	\section{}
	In this appendix we recall from Di Nunno et al. \cite{di2009} the Clark-Ocone formula and the chain rule in the Brownian and Poisson probability space $(\Omega,\mathcal{F},\mathbb{P})$.
	\begin{theorem}
		Let $F\in \mathbb{D}^{1,2}$. Then
		$$
			F=\mathbb{E}[F]	+ \int_0^T\mathbb{E}[D_tF|\mathcal{F}_t]dW(t)+\int_0^T\int_{\mathbb{R}_0}\mathbb{E}[D_{t,\zeta}F|\mathcal{F}_t]\tilde{N}(ds,d\zeta).	
		$$
		
	\end{theorem}
	\begin{theorem}[Chain Rule]
			Let $F=F_1,\ldots,F_m\in \mathbb{D}_{1,2}$, and let $\phi:\mathbb{R}^m\longrightarrow \mathbb{R}$ be a bounded continuously differentiable. Then
			$$
			D_{t,\zeta}\varphi(F)=\varphi(F +D_{t,\zeta}F)-\varphi(F).
			$$
		 $$
		 D_t\phi (F)=\sum_{j=1}^m\frac{\partial}{\partial x}\phi(F_1,\ldots,F_m)D_tF_j\qquad dt\times d\mathbb {P} - a.e.
		 $$
	\end{theorem}

\end{document}